\documentclass[runningheads, envcountsame, a4paper]{llncs} 

\usepackage[hyphens]{url}
\usepackage{microtype}
\usepackage{amsmath}
\usepackage{amssymb}
\usepackage{amsfonts}
\usepackage{breqn} 

\DeclareMathOperator{\Suff}{Suff}
\DeclareMathOperator{\Pref}{Pref}
\DeclareMathOperator{\Fact}{Fact}

\DeclareMathOperator{\PAL}{PAL}

\renewcommand{\epsilon}{\varepsilon}

\newcommand{\LL}{\mathcal{L}}

\begin{document}

\title{On the Number of Closed Factors in a Word}

\sloppy

\date{\today}

\author{Golnaz Badkobeh\inst{1} \and Gabriele Fici\inst{2} \fnmsep \thanks{Partially supported by Italian MIUR Project PRIN~2010LYA9RH, ``Automi e Linguaggi Formali: Aspetti Matematici e Applicativi''.} \and Zsuzsanna Lipt\'ak\inst{3}}

\institute{
Department of Computer Science, University of Sheffield, UK\\ \email{g.badkobeh@sheffield.ac.uk}  \and
Dipartimento di Matematica e Informatica, Universit\`a di Palermo, Italy\\
\email{gabriele.fici@unipa.it} \and
Dipartimento di Informatica, Universit\`a di Verona, Italy \\
\email{zsuzsanna.liptak@univr.it}
}

\toctitle{On the Number of Closed Factors in a Word}
\tocauthor{Golnaz~Badkobeh, Gabriele~Fici, and Zsuzsanna~Lipt\'ak}
\maketitle

\setcounter{footnote}{0}

\begin{abstract}
A closed word (a.k.a. periodic-like word or complete first return) is a word whose longest border does not have internal occurrences, or, equivalently, whose longest repeated prefix is not right special. We investigate the structure of closed factors of words. We show that a word of length $n$ contains at least $n+1$ distinct closed factors, and characterize those words having exactly $n+1$ closed factors. Furthermore, we show that a word of length $n$ can contain $\Theta(n^{2})$ many distinct closed factors.
\end{abstract}

\section*{Introduction}

It is  known (see for example \cite{DrJuPi01}) that any word $w$ of length $n$ contains at most $n+1$ palindromic factors.
Triggered by this result, several researchers initiated a study to characterize words
that can accommodate a maximal number of palindromes, called \emph{rich} (or \emph{full}) \emph{words} (see, for example, \cite{BrHaNiRe04,GlJuWiZa09,BuDelDel09,BuDelGlZa09,ReRo09}).

In this paper, we consider the notion of \emph{closed word} (a.k.a.~periodic-like word or complete first return). A word $w$ is  closed if and only if it is empty or has a factor $v\neq w$ occurring exactly twice in $w$, as a prefix and as a suffix of $w$. We also say in this case that $w$ is a complete return to $v$.
For example, $aaa$, $ababa$, $ccabcc$ are all  closed words (they are complete returns to $aa$, $aba$ and $cc$, respectively), while $ab$ and $abaabab$ are not. As shown in Proposition \ref{prop:exp}, any word whose exponent is at least two is closed.

The  \emph{closed factors} of a word are its factors that are  closed words. In contrast to the case of palindromic factors, we show that a word of length $n$ contains at least $n+1$  closed factors (Lemma \ref{lem:bound}). Inspired by this property, we study the class of words that contain the smallest number of  closed factors, and we call them \emph{CR-poor words}.

As an example, $abca$ is a CR-poor word, since it has length $4$ and exactly $5$  closed factors, namely $\epsilon, a, b, c$ and $abca$, whereas the word $ababa$ is not CR-poor since it has length $5$ but contains  $8$  closed factors: $\epsilon$, $a$, $b$, $aba$, $bab$, $abab$, $baba$ and $ababa$. 

However, there is some relation between rich words and CR-poor words. Bucci, de Luca and De Luca~\cite{BuDelDel09} showed that a palindromic word is rich if and only if all of its palindromic factors are  closed. We show, in Proposition \ref{prop:rich}, that if a word $w$ has the property that all of its  closed factors are palindromes, then $w$ is a CR-poor word, and it is also rich. CR-poor words are also connected to some problems on \emph{privileged words} (see \cite{Pel13}).

 While having only palindromic  closed factors is a necessary and sufficient condition for a binary word to be CR-poor (Theorem \ref{theor:mainbin}), we prove that in a word $w$ over an alphabet $\Sigma$ of arbitrary cardinality, the set of  closed factors and the set of palindromic factors of $w$ coincide if and only if $w$ is both rich and CR-poor (Proposition \ref{prop:richpoor}).

In Theorem \ref{theor:main}, we give a combinatorial characterization of CR-poor words over an alphabet $\Sigma$ of cardinality greater than one: A word over $\Sigma$ is CR-poor if and only if it does not contain any  closed factor that is a  complete return to $xy$, for $x,y$ different letters in $\Sigma$. In other words, CR-poor words are exactly those words having as their closed factors only  complete returns to powers of a single letter. 
As a consequence, the language of CR-poor words over $\Sigma$ is a regular language. In contrast, the language of  closed words is not regular (Proposition \ref{prop:reg}).

We give some further characterizations of CR-poor words in the case of the binary alphabet (Theorem \ref{theor:mainbin}). One of them is that the binary CR-poor words are the \emph{bitonic words}, i.e., the conjugates to words in $a^{*}b^{*}$. We therefore have that binary CR-poor words form a regular subset of the language of rich words.

Finally, we show that a word of length $n$ can contain $\Theta(n^{2})$ many distinct closed factors (Theorem \ref{max}).

\section{Closed Words}

A \textit{word} is a finite sequence of elements from a finite set $\Sigma$. We refer to the elements of $\Sigma$ as \emph{letters} and to $\Sigma$ as the \emph{alphabet}.   The $i$-th letter of a word $w$ is denoted by $w_{i}$. Given a word $w=w_1w_2\cdots w_n$, with $w_i\in\Sigma$ for $1\leq i\leq n$, the nonnegative integer $n$ is the \emph{length} of $w$, denoted by $|w|$. The empty word has length zero and is denoted by $\varepsilon$. 
The set of all words over $\Sigma$ is denoted by $\Sigma^*$.  
Any subset of $\Sigma^*$ is called a \emph{language}. A language is \emph{regular} (or \emph{rational}) if it can be recognized by a finite state automaton.

A \emph{prefix} (resp.~a \emph{suffix}) of a word $w$ is any word $u$ such that $w=uz$ (resp.~$w=zu$) for some word $z$. A \emph{factor} of $w$ is a prefix of a suffix (or, equivalently, a suffix of a prefix) of $w$. The set of prefixes, suffixes and factors of the word $w$ are denoted  by $\Pref(w)$, $\Suff(w)$ and $\Fact(w)$ respectively. A \emph{border} of a word $w$ is any word in $\Pref(w)\cap \Suff(w)$ different from $w$. From the definitions, we have that $\epsilon$ is a prefix, a suffix, a border and a factor of any word.  
An \emph{occurrence} of a factor $u$ in $w$ is a factorization $w = vuz$. An occurrence of $u$ is \emph{internal} if both $v$ and $z$ are non-empty. 

The word $\tilde{w}=w_{n}w_{n-1}\cdots w_{1}$ is called the \emph{reversal} (or \emph{mirror image}) of $w$. A \emph{palindrome} is a word $w$ such that $\tilde{w}=w$.
In particular, the empty word is  a palindrome. A \emph{conjugate} of a word $w$ is any word of the form $vu$ such that $uv=w$, for some $u,v\in \Sigma^{*}$. A conjugate of a word $w$ is also called a \emph{rotation} of $w$. 

A \emph{period} for the word $w$ is a positive integer $p$, with $0<p\leq |w|$, such that
$w_{i}=w_{i+p}$ for every $i=1,\ldots ,|w|-p$. Since $|w|$ is always a period for $w$, we have that every non-empty word has at least one period. We can unambiguously define \textit{the} period of the word $w$ as the smallest of its periods.  
The \emph{exponent} of a word $w$ is the ratio between its length and its smallest period. A \emph{power} is a word whose exponent is an integer greater than $1$. A word that is not a power is called \emph{primitive}

We denote by $\PAL(w)$ the set of factors of $w$ that are palindromes. A word $w$ of length $n$ is \emph{rich}~\cite{GlJuWiZa09} (or \emph{full}~\cite{BrHaNiRe04})  if $|\PAL(w)|= n+1$, i.e., if it contains the largest number of palindromes a word of length $n$ can contain.

A language $L$ is called \emph{factorial} if $L=\Fact(L)$, i.e., if $L$ contains all the factors of its words. A language $L$ is \emph{extendible} if for every word $w\in L$, there exist letters $a,b\in \Sigma$ such that $awb\in L$. The language of rich words over a fixed alphabet $\Sigma$ is an example of a factorial and extendible language.

We recall the definition of  closed word given in \cite{Fi11}:

\begin{definition}\label{def: closed}
A word $w$ is \emph{closed} if and only if it is empty or has a factor $v\neq w$ occurring exactly twice in $w$, as a prefix and as a suffix of $w$.
\end{definition}

The word $aba$ is a  closed, since its factor $a$ appears in it only as a prefix and as a suffix. The word $abaa$, on the contrary, is not  closed. Note that for any letter $a\in \Sigma$ and for any integer $n>0$, the word $a^{n}$ is  closed, $a^{n-1}$ being a factor occurring only as a prefix and as a suffix in it (this includes the special case of single letters, for which $n=1$ and $a^{n-1}=\epsilon$). 

\begin{remark}
The notion of  closed word is equivalent to that of \emph{periodic-like} word \cite{CaDel01a}. A word $w$ is periodic-like if its longest repeated prefix does not have two occurrences in $w$ followed by different letters, i.e., if its longest repeated prefix is not right special.

The notion of  closed word is also closely related to the concept of \emph{complete return} to a factor, as considered in \cite{GlJuWiZa09}. A complete return to the factor $u$ in a word $w$ is any factor of $w$ having exactly two occurrences of $u$, one as a prefix and one as a suffix. Hence a non-empty word $w$ is  closed if and only if it is a complete return to one of its factors; such a factor is clearly both the longest repeated prefix and the longest repeated suffix of $w$ (i.e., the longest border of $w$). 
\end{remark}
 
\begin{remark}\label{obs}
Let $w$ be a non-empty word over $\Sigma$. The following characterizations of  closed words follow easily from the definition:
 
\begin{enumerate}
 \item $w$ has a factor $v\neq w$ occurring exactly twice in $w$, as a prefix and as a suffix of $w$; 
 \item the longest repeated prefix (resp.~suffix) of $w$ does not have internal occurrences in $w$, i.e., occurs in $w$ only as a prefix and as a suffix;
  \item the longest repeated prefix (resp.~suffix) of $w$ does not have two occurrences in $w$ followed (resp.~preceded) by different letters;
 \item $w$ has a border that does not have internal occurrences in $w$;
 \item the longest border of $w$ does not have internal occurrences in $w$;
 \item $w$ is a complete return to its longest repeated prefix;
 \item $w$ is a complete return to its longest border.
\end{enumerate}
\end{remark}

For more details on  closed words and related results see~\cite{CaDel01a,BuDelDel09,Fi11,BuDelFi13,DelFi13,BaBa14,slo}.

We end this section by exhibiting some properties of closed words.

\begin{proposition}\label{prop:exp}
Any word whose exponent is at least $2$ is closed.
\end{proposition}

\begin{proof}
Let $w=v^{n}v'$ for $n\geq 2$, $v$ a primitive word, and $v'$ a prefix of $v$ such that the exponent of $w$ is equal to $n+|v'|/n$. Then $v^{n-1}v'$ is a border of $w$. If $v^{n-1}v'$ has an internal occurrence in $w$, then there exists a proper prefix $u$ of $v$ such that $uv=vu$, and it is a basic result in Combinatorics on Words that two words commute if and only if they are powers of the same word, in contradiction with our hypotheses on $u$ and $v$.
 \qed
\end{proof}

Moreover, it is easy to see that for any rational number $x$ between $1$ and $2$, there exists a closed word having exponent $x$ (it is sufficient to take a word over $\{a,b\}$ ending with $b$ and with only one other occurrence of $b$, placed in the first half of the word).

\begin{proposition}\label{prop:reg}
Let $\Sigma$ be an alphabet of cardinality $|\Sigma|\ge 2$. The language of  closed words over $\Sigma$ is not regular.
\end{proposition}

\begin{proof}
Let $L$ be the language of  closed words over $\Sigma$ and let $a,b\in \Sigma$ be different letters. Let us assume that $L$ is regular. This implies that also $L\cap a^{*}b^{*}a^{*}$ is regular, since $a^{*}b^{*}a^{*}$ is a regular language and the intersection of two regular languages is regular. We claim that $L\cap a^{*}b^{*}a^{*}=\{a^{n}b^{m}a^{n} \mid n,m\ge 0\}$, which is not a regular language, and so we have a contradiction.

Clearly, every word in $\{a^{n}b^{m}a^{n} \mid n,m\ge 0\}$ is  closed. Suppose now that $w$ belongs to $a^{*}b^{*}a^{*}$. Hence, $w=a^{n}b^{m}a^{k}$, for some $n,m,k\ge 0$. If $n\neq k$, say $n<k$, then the longest repeated prefix of $w$ is $a^{n}$ and it has at least one internal occurrence in $w$. By Remark \ref{obs}, $w$ is not  closed. The case $n>k$ is symmetric. 
\qed
\end{proof}

Finally, we recall two results from \cite{DelFi13}.

\begin{lemma}\cite[Lemma 4]{DelFi13}\label{lem:4}
 Let $w$ be a non-empty word over $\Sigma$. Then there exists at most one letter $x\in \Sigma$  such that $wx$ is closed.
\end{lemma}

\begin{lemma}\cite[Lemma 5]{DelFi13}\label{lem:5}
 Let $w$ be a closed word. Then $wx$, $x\in \Sigma$, is closed if and only if $wx$ has the same period of $w$.
\end{lemma}

\section{Closed Factors}

Let $w$ be a word. A factor of $w$ that is a  closed word is called a \emph{closed factor} of $w$.  The set of  closed factors of the word $w$ is denoted by $C(w)$. 

\begin{lemma}\label{lem:bound}
For any word $w$ of length $n$, one has $|C(w)|\ge n+1$.
\end{lemma}

\begin{proof}
We show that every position of $w$ is the ending position of an occurrence of a distinct  closed factor of $w$. 
Thus $w$ contains at least $n$ non-empty  closed factors, and the claim follows. Indeed, let $v$ be the longest non-empty  closed factor ending in position $i$, so that $w_{i-|v|+1}\cdots w_{i} =v$. Since $a$ is  closed for every $a\in \Sigma$, such a factor always exists.  If $v$ did not occur before in $w$, then we are done. Otherwise, let $j$ be the largest position smaller than $i$ such that $w_{j-|v|+1}\cdots w_{j} =v$. Set $v' = w_{j-|v|+1}\cdots w_i$ and observe that $v'$ is a  closed factor ending in $i$, with longest border $v$.  But $|v'| > |v|$, in contradiction to the choice of $v$.\qed
\end{proof}

\begin{lemma}\label{lem:trian}
 For any words $u,v$ one has $|C(u)|+|C(v)|\le |C(uv)|+1$.
\end{lemma}

\begin{proof}
Clearly, $C(u)\subseteq C(uv)$. In order to prove the statement, it is sufficient to prove that for any non-empty $z$ in $C(v)$, there exists an $f(z)$ in $C(uv)\setminus C(u)$ and $f$ is injective. So let $z\in C(v)$, $uv=w=w_{1}\cdots w_{n}$, and let $j$ be the smallest integer greater than $|u|$ such that $z=w_{j}\cdots w_{j+|z|-1}$. If $j$ is the smallest integer such that $z=w_{j}\cdots w_{j+|z|-1}$, then set $f(z)=z$. Otherwise, there is in $w$ a  closed $z'$ to $z$ ending in position $w_{j+|z|-1}$. If this is the first occurrence of $z'$ in $w$, then set $f(z)=z'$, otherwise repeat the construction  for $z'$. Eventually, we will find a  closed factor $f(z)=z^{(k)}$  whose first occurrence in $w$ ends in position $w_{j+|z|-1}$. 

By construction, $f$ has the desired properties. 
\qed
\end{proof}

\begin{proposition}\label{prop:rich}
 Let $w$ be a word of length $n$. If $C(w)\subseteq \PAL(w)$, then $C(w)=\PAL(w)$ and $|C(w)|=|\PAL(w)|= n+1$. In particular,  $w$ is a rich word.
\end{proposition}

\begin{proof}
On the one hand, from Lemma \ref{lem:bound}, one has $|C(w)|\ge n+1$. On the other hand, one has $|\PAL(w)|\le n+1$. Hence, if $C(w)\subseteq \PAL(w)$, then it must be $C(w)=\PAL(w)$ and  $|C(w)|=|\PAL(w)|= n+1$, and so $w$ is a rich word.
\qed
\end{proof}

Bucci et al.~\cite[Proposition 4.3]{BuDelDel09} showed  that a word $w$ is rich if and only if every  closed factor $v$ of $w$ has the property that the longest palindromic prefix (or suffix) of $v$ is unrepeated in $v$. Moreover, they proved the following remarkable result:

\begin{theorem}[Bucci et al.~{\cite[Corollary 5.2]{BuDelDel09}}]
\label{theor:buc}
 A palindromic word $w$ is rich if and only if $\PAL(w)\subseteq C(w)$. 
\end{theorem}

In Section \ref{sec:bin}, we will prove that the condition $\PAL(w)= C(w)$ characterizes the  CR-poor words over a binary alphabet.

\section{CR-poor Words}

By Lemma \ref{lem:bound}, we have that $n+1$ is a lower bound on the number of  closed factors of a word of length $n$. We  introduce the following definition:

\begin{definition}
 A word $w\in \Sigma^{*}$ is \emph{CR-poor} if $|C(w)|=|w|+1$. We also set $$\LL_{\Sigma}=\{w\in \Sigma^{*} :  |C(w)|=|w|+1\}$$ the language of CR-poor words over the alphabet $\Sigma$.
\end{definition}

\begin{remark}
 If $|\Sigma|=1$, then $\LL_{\Sigma}=\Sigma^{*}$. So in what follows we will suppose $|\Sigma|\ge 2$.
\end{remark}

Note that, for any alphabet $\Sigma$, the language $\LL_{\Sigma}$ of CR-poor words over $\Sigma$ is closed under reversal. Indeed, it follows from the definition that a word $w\in \Sigma^{*}$ is  closed if and only if its reversal $\tilde{w}$ is  closed. 
 
\begin{proposition}\label{prop:fac}
The language $\LL_{\Sigma}$ of CR-poor words over $\Sigma$ is a factorial language.
\end{proposition}

\begin{proof}
We have to prove that for any word CR-poor $w$ and any factor $v$ of $w$, $v$ is a CR-poor word. Suppose by contradiction that there exists a CR-poor word  $w$ containing a factor $v$ that is not a CR-poor word, i.e., $w \in \LL_{\Sigma}$, $w=uvz$ and $|C(v)|>|v|+1$. By Lemma  \ref{lem:trian}, $|C(w)|\ge |C(u)|+|C(v)|+|C(z)|-2 > |u|+|z|+|v|+1=|w|+1$ and therefore $w$ cannot be a CR-poor word.
\qed
\end{proof}

The following technical lemma will be used in the proof of the next theorem.

\begin{lemma}\label{lem:tech}
Let $w$ be a CR-poor word over the alphabet $\Sigma$ and $x\in \Sigma$. The word $wx$ (resp.~$xw$) is CR-poor if and only if it has a unique suffix (resp.~prefix) that is  closed and is not a factor of $w$.  
\end{lemma}

\begin{proof}
We prove the statement for $wx$, the one for $xw$ will follow by symmetry. The ``if'' part is straightforward. 
For the ``only if'' part, recall from the proof of Lemma~\ref{lem:bound} that there is at least one new  closed factor ending in every position, so in particular $wx$ has at least one suffix that is closed and is not a factor of $w$. 
\qed
\end{proof}

\begin{remark}\label{rem:cr}
Suppose that a word $w$ contains as a factor a complete return to some word $u$. Then for every factor $u'$ of $u$, the word $w$ contains as a factor a complete return to $u'$.
\end{remark}

We now give a characterization of CR-poor words.

\begin{theorem}\label{theor:main}
A word $w$ over $\Sigma$ is CR-poor if and only if for any two different letters $a,b\in \Sigma$, $w$ does not contain any complete return to $ab$. In other words, 
$$\LL_{\Sigma} = \Sigma^* \setminus \bigcup_{a\neq b}\ \Sigma^*ab\Sigma^*ab\Sigma^*.$$
\end{theorem}

\begin{proof}
Let $u$ be a complete return to $ab$ for $a,b\in \Sigma$ different letters. We claim that $u$ is not CR-poor. Since by Proposition \ref{prop:fac}, a CR-poor word cannot contain a factor that is not CR-poor, once the claim is proved the ``only if'' part of the theorem follows. So let $u'$ be the longest suffix of $u$ that is closed and starts with the letter $b$. Such a suffix exists since $u$ contains at least two occurrences of $b$. Then $u'$ is unioccurrent in $u$, and since $u$ is a closed suffix of itself we have, by Lemma \ref{lem:tech}, that $u$ is not CR-poor.

Conversely, suppose that the word $w$ is not CR-poor. Then, analogously as in the proof of Lemma 3, it follows that there is a position $i$ of $w$ such that there are at least two different closed factors $u$ and $u'$ of $w$ that end in position $i$ and do not occur in $w_{1}\cdots w_{i-1}$.
If both $u$ and $u'$ are complete returns to a power of the letter $w_{i}$, then one of them must occur in $w_{1}\cdots w_{i-1}$, so this situation is not possible, and we can therefore suppose that there is a factor ending in position $i$ that is a complete return to a word containing at least two different letters. The statement then follows from Remark \ref{rem:cr}.
\qed
\end{proof}

\begin{corollary}\label{cor:powersingle}
A word $w$ over $\Sigma$ is CR-poor if and only if every  closed factor of $w$ is a complete return to a power of a single letter.
\end{corollary}

\begin{corollary}
The language $\LL_{\Sigma}$ of CR-poor words over $\Sigma$ is a regular language. 
\end{corollary}

We can now state the following result:

\begin{proposition}\label{prop:richpoor}
 Let $w$ be a word over $\Sigma$. Then $C(w)=\PAL(w)$ if and only if $w$ is rich and CR-poor. 
\end{proposition}

\begin{proof}
If $C(w)=\PAL(w)$, then $|C(w)|=|\PAL(w)|$, and since $|C(w)|\ge |w|+1$ (by Lemma \ref{lem:bound}) and $|\PAL(w)|\le |w|+1$, then it must be $|C(w)|=|\PAL(w)|=|w|+1$, and hence by definition $w$ is rich and CR-poor.

Conversely, suppose that $w$ is rich and CR-poor. Let $v\in C(w)$. By Corollary \ref{cor:powersingle}, $v$ is a  complete return to a power of a single letter, so $v$ is a complete return to a palindrome.
It is known (see \cite[Theorem 2.14]{GlJuWiZa09}) that a word is rich if and only if all of its factors that are  complete returns to a palindrome are palindromes themselves. Therefore, $v$ is a palindrome, and hence we proved that $C(w)\subseteq \PAL(w)$.
By Proposition \ref{prop:rich}, $C(w)=\PAL(w)$ and we are done.
\qed
\end{proof}

\section{The Case of Binary Words}\label{sec:bin}

In this section we fix the alphabet $\Sigma=\{a,b\}$. For  simplicity of exposition, we will denote the language of CR-poor words over $\{a,b\}$ by $\LL$ rather than by $\LL_{\{a,b\}}$.
We first recall the definition of bitonic word.

\begin{definition}
A word $w\in \{a,b\}^{*}$ is \emph{bitonic} if it is a conjugate of a word in $a^{*}b^{*}$, i.e., if it is of the form $a^{i}b^{j}a^{k}$ or $b^{i}a^{j}b^{k}$ for some integers $i,j,k\ge 0$.
\end{definition}

By Theorem \ref{theor:main}, it is easy to see that a binary word is in $\LL$ if and only if it is bitonic.

\begin{lemma}\label{lem:bit}
Let $w$ be a bitonic word. Then $C(w)\subseteq \PAL(w)$. 
\end{lemma}

\begin{proof}
Since $w$ is bitonic, a  closed factor of $w$ can only be the complete return to a power of a single letter. So a  closed factor $u$ of $w$ is of the form $u=a^{n}$, $u=b^{n}$, $u=a^{n}b^{m}a^{n}$ or $u=b^{n}a^{m}b^{n}$, for some $n,m>0$, and these words are all palindromes.
\qed
\end{proof}

Thus, by Proposition \ref{prop:rich}, any bitonic word $w$ of length $n>0$ contains exactly $n+1$  closed factors and so is a CR-poor word. We therefore have the following characterizations of CR-poor binary words.

\begin{theorem}\label{theor:mainbin}
 Let $w\in \{a,b\}^{*}$. The following are equivalent:
 \
\begin{enumerate}
 \item $w\in \LL$;
  \item $w$ does not contain any complete return to $ab$ or $ba$;
 \item $C(w)\subseteq\PAL(w)$;
  \item $C(w)=\PAL(w)$;
 \item $w$ is a bitonic word.
\end{enumerate}
\end{theorem}

Notice that the condition $C(w)\subseteq\PAL(w)$ does not hold in general for CR-poor words over alphabets larger than two. As an example, the word $abca$ is CR-poor but contains a  closed factor ($abca$) that is not a palindrome.  In view of Theorem \ref{theor:buc}, a natural question would be that of establishing whether a palindrome $w$ is CR-poor if and only if $C(w)= \PAL(w)$, i.e., whether the characterization in Theorem \ref{theor:mainbin} can be generalized to larger alphabets at least for palindromes. However, the answer to this question is negative since, for example, the word $w=abcacba$ is a CR-poor palindrome and contains the non-palindromic  closed factor $abca$. Note that, coherently with Theorem \ref{theor:buc} (and with Proposition \ref{prop:richpoor}), $w$ is not rich. 
However, in the case of a  binary alphabet, we have, by Theorem \ref{theor:mainbin} and Proposition \ref{prop:rich}, that every CR-poor word is rich. Since by Theorem~\ref{theor:main} it follows that the language $\LL_{\Sigma}$ is extendible for any alphabet $\Sigma$, the language $\LL$ is therefore a factorial and extendible subset of the language of (binary) rich words.

In the following proposition we exhibit a  closed enumerative formula for the language $\LL$.

\begin{proposition}\label{prop:enum}
 For every $n>0$, there are exactly $n^{2}-n+2$ distinct words in $\LL$.
\end{proposition}

\begin{proof}
Each of the $n-1$ words of length $n>0$ in $a^{+}b^{+}$ has $n$ distinct rotations, while for the words $a^{n}$ and $b^{n}$ all the rotations coincide. Thus, there are $n(n-1)+2$ bitonic words of length $n$, and the statement follows from Theorem \ref{theor:mainbin}.
\qed
 \end{proof}

\section{How Many Closed Factors Can a Word Contain?}

We showed in Lemma \ref{lem:bound} that any word of length $n$ contains at least $n+1$ distinct closed factors. But how many closed factors, at most, can a word contain? We provide an answer in the following theorem.

\begin{theorem}\label{max}
For every $n > 4$, there exists a word $w \in \{a,b\}^n$ with quadratically many closed factors.
\end{theorem}

\begin{proof}
Let $n>4$ be fixed. We construct a word $w$ of length $n$ such that
$|C(w)| \geq (k+1)(k+2)/2$, where $k=\lfloor n/4 \rfloor$.

Let $w=a^kb^ka^kb^ka^{n-4k}$. Clearly $|w|=n$.
Let $v_{i,j}=w_{i}\cdots w_{j}$, $1\le i\le j\le n$, be a factor of $w$. We claim that for every $i=1,2,\ldots,k-1$, every factor $v_{i,j}$, with $3k-1+i \le j\le 4k$, is closed. Indeed, fixed $i$ between $1$ and $k-1$, the factor $v_{i,3k-1+i}$, of length $3k$, is equal to $a^{k-i+1}b^{k}a^{k}b^{i-1}$, and therefore it is closed since it is a complete return to $a^{k-i+1}b^{i-1}$. Then, for every $j$ such that $3k-1+i \le j\le 4k$, the factor $v_{i,j}$ has the same period of $v_{i,3k-1+i}$, and therefore is closed by Lemma \ref{lem:5}. 

Finally, notice that whenever $(i',j')$ is different from $(i,j)$, for $i'$ and $j'$ in the same range of $i$ and $j$, respectively (that is, $1\le i\le k-1$ and $3k-1+i \le j\le 4k$), the factor $v_{i',j'}$ is different from the factor $v_{i,j}$. 

Therefore we conclude that $w$ contains at least $(k+1)(k+2)/2 = \Theta(n^2)$ many different closed factors, and we are done. 
 \qed
 \end{proof}

It is possible to exhibit a formula for the precise value of the maximal number of closed factors in a word of length $n$, but we think this adds nothing to the general picture provided by Theorem \ref{max}. Moreover, the words realizing the upper bound do not have a nice characterization, contrarily to the case of words realizing the lower bound, discussed in the previous sections. However, for completeness, we report in Table \ref{tab:max} the first values of the sequence of the maximum number of closed factors for binary words. 

\begin{table}
\centering  
\begin{small}
\begin{raggedright}
\begin{tabular}{c *{30}{@{\hspace{2.1mm}}l}}
 $n$  & 1\hspace{1ex} & 2\hspace{1ex} & 3\hspace{1ex} & 4\hspace{1ex} & 5\hspace{1ex} & 6\hspace{1ex} & 7\hspace{1ex} &
8\hspace{1ex} & 9\hspace{1ex} & 10 & 11 & 12 & 13 & 14 & 15 & 16 & 17 & 18 & 19 & 20 \\
\hline \rule[-6pt]{0pt}{22pt}
$ max$ & 2& 3& 4& 6& 8& 10& 12& 15& 18& 21& 25& 29& 33 & 37 & 42 & 47& 52 & 58 & 64 & 70\\
\hline \\
\end{tabular}
\end{raggedright}\caption{\label{tab:max}The sequence of the maximum number of closed factors in a binary word.}
\end{small}
\end{table}

\section{Conclusion and Open Problems}

This paper is a first attempt to study the set of closed factors of a finite word. In particular, we investigated the words with the smallest number of closed factors, which we referred to as CR-poor words. We provided a combinatorial characterization of these words and exhibited some relations with rich words.

An enumerative formula for rich words is not known, not even in the binary case. A possible approach to this problem is to separate rich words in subclasses to be enumerated separately. Our enumerative formula for binary CR-poor words given in Proposition \ref{prop:enum} could constitute a step towards this direction.

The set of closed factors could be investigated for specific (finite or infinite) words or classes of words, and could be a tool to derive new combinatorial properties of words.

Finally, the notion of closed factor has recently found applications in string algorithms \cite{BaBa14}, hence a better understanding of the structure of closed factors of a word could lead to some applications.

\bibliographystyle{splncs03}
\bibliography{closed}

\end{document}